\documentclass{LMCS}

\usepackage{latexsym}
\usepackage{amssymb}
\usepackage{bussproofs}
\usepackage{tikz}
\usepackage{enumerate,hyperref}

\usepackage{fancyvrb}
\DefineVerbatimEnvironment{code}{Verbatim}{}
\newcommand{\ignore}[1]{}

\input{sdmacros.sty}

\def\doi{7 (1:8) 2011}
\lmcsheading%
{\doi}
{1--24}
{}
{}
{Jun.~11, 2010}
{Mar.~23, 2011}
{}   

\begin{document} 

\title{From coinductive proofs to exact real arithmetic: 
theory and applications} 
\author{Ulrich Berger}
\address{Swansea University, Swansea SA2 8PP,
  Wales, UK}
\email{u.berger@swansea.ac.uk}

\begin{abstract}
Based on a new coinductive characterization of continuous
functions we extract certified programs for exact real number computation
from constructive proofs. The extracted programs construct and combine
exact real number algorithms with respect to the binary 
signed digit representation of real numbers.
The data type corresponding to the coinductive definition of 
continuous functions consists of finitely branching non-wellfounded
trees describing when the algorithm writes and reads digits.
We discuss several examples including the extraction of programs for
polynomials up to degree two and the definite integral
of continuous maps.
This is a revised and substantially extended version of the conference 
paper~\cite{BergerCSL09}.
%
%This is a pilot study in using proof-theoretic methods for %obtaining
%certified algorithms in exact real arithmetic.
\end{abstract}

\keywords{Proof theory, realizability, program extraction, 
coinduction, exact real number computation} 

\subjclass{??}

\maketitle

\section{Introduction}
\label{sec-intro}
Most of the recent work on exact real number
computation describes algorithms for functions on certain
exact representations of the reals (for example streams of
signed digits 
\cite{EscardoMarcial-Romero07,GeuversNiquiSpittersWiedijk07}
 or linear fractional transformations
\cite{EdalatHeckmann02}) and 
proves their correctness using a certain proof method 
(for example 
coinduction~\cite{CiaffaglioneGianantonio06,Bertot07,BergerHou07,Niqui08}). 
Our work has a similar aim, and builds on the work cited above, 
but there are two important differences.
The first is \emph{methodological}: we do not `guess' an algorithm
and then verify it, instead we \emph{extract} it from a proof, by
some (once and for all) proven correct method.
That this is possible in principle is well-known. Here we want to make the
case that it is also feasible, and that interesting and nontrivial
new algorithms can be obtained
(see also \cite{Schwichtenberg08,BergerSeisenberger05} for related work
on program extraction in constructive analysis and inductive definitions).
The second difference is \emph{algorithmic}: our method represents a
uniformly continuous real function not by a \emph{function} operating
on representations of reals, but by an infinite \emph{tree} that
contains information not only about the real function as a point map,
but also about its modulus of continuity.  Since the representing tree
is a pure data structure (without function component) a lazy
programming language, like Haskell, will memoize computations which
improves performance in certain situations.
%A similar representation of stream transformers 
%has been studied in~\cite{Ghanietal09}.

%We show how to extract from constructive proofs
%tree structures that represent algorithms for
%continuous real functions defined on a compact interval
%w.r.t.\ the signed digit representation of real numbers.

A crucial ingredient in the proofs (that we use for program extraction)
is a coinductive definition of the notion of uniform continuity  (u.\ c.).
Although, classically, continuity and uniform continuity
are equivalent for functions defined on a compact interval (we only consider
such functions), 
it is a suitable constructive
definition of \emph{uniform} continuity which matters for our purpose.
For convenience, we consider as domain and range of our functions
only the interval 
$\II : =[-1,1] = \set{x\in\RR\mid |x|\le 1} $
and, for the purpose of this introduction, only unary functions.
However, later we will also look at functions
of several variables where one has to deal with the non-trivial 
problem of choosing the input streams from which 
the next digit is consumed, a choice which can have
a big influence on the performance of the program.

We let $\SD := \set{-1,0,1}$
be the set of (binary) \emph{signed digits}.
By $\SDS$ we denote the set of all infinite streams 
$a = a_0 : a_1 :a_2 : \ldots$ of signed digits $a_i\in\SD$.
A signed digit stream $a\in \SDS$ represents the real number
\[\sdval{a} := \sum_{i\ge 0} a_i2^{-(i+1)}\in\II\]
A function $f : \II \to \II$ is \emph{represented} by a stream transformer
$\hat{f}:\SDS \to \SDS$ if $f\comp\sigma = \sigma\comp\hat{f}$.
%
%In Sect.~\ref{sec-coco} 
The coinductive definition of uniform continuity, given in 
Sect.~\ref{sec-realizability}, allows us to extract 
from a constructive proof of the u.\ c.\ of a function 
$f : \II \to\II$ an algorithm for a stream transformer $\hat{f}$ 
representing $f$. 
Furthermore, we show directly and constructively that the coinductive
notion of u.\ c.\ is closed under composition.
The extracted stream transformers are represented
by finitely branching non-wellfounded trees which, if executed in a lazy 
programming language, give rise to memoized algorithms.
These trees turn out to be closely related to the data structures
studied in \cite{Ghanietal09,Ghanietal09a}, and the
extracted program from the proof of closure under composition
is a generalization of the tree composing program defined there.

In Sect.~\ref{sec-ind-coind}, we briefly review inductive and coinductive
sets defined by monotone set operators. We give some simple examples,
among them a characterization of the real numbers in the
interval $\II$ by a coinductive predicate $\coco_0$. 
The method of program extraction from proofs involving induction and
coinduction is discussed informally, but in some detail, in 
Sect.~\ref{sec-realizability}. The earlier examples 
are continued and a program transforming fast Cauchy
representations into signed digit representations is extracted from
a coinductive proof. 
%
%Wellfounded induction and its particularly simple computational
%interpretation are discussed in Sect.~\ref{sec-wf}.
%
In Sect.~\ref{sec-coco}, the coinductive characterization $\coco_0$ of 
real numbers is generalized to nested coinductive/inductive predicates 
$\coco_n$ characterizing uniformly continuous real functions of $n$ 
arguments, and closure under composition is proven.
In Sect.~\ref{sec-digit}, we study wellfounded induction from the perspective
of program extraction and introduce the notion of a 
\emph{digital system} as a technical tool for showing that certain families
of functions are contained in $\coco_n$. The positive effect of memoization is
demonstrated by a case study on iterated logistic maps (which are special 
polynomials of degree 2). Furthermore, we prove that the predicates $\coco_n$
capture precisely uniform continuity. 
In Sect.~\ref{sec-integration} we extract a program for integration
from a proof that the definite integral on $\II$ of a function in 
$\coco_1$ can be approximated by rational numbers with any given precision.

% \paragraph{Notations.}
% By $\pow{X}$ we denote the powerset of a set $X$, 
% and by $X\to Y$ (sometimes also written $Y^X$) the set
% of all functions from $X$ to $Y$. 
% %
% %$X^n:= X\times\ldots\times X$ ($n$ times $X$).
% %
% $X\to Y\to Z$ is shorthand for $X\to(Y\to Z)$. We will sometimes
% write $X(x)$ instead of $x\in X$ and
% $f\colon X\to Y$ instead of $f\in X\to Y$. If
% $f\colon X\to Y$ and $Z$ is a set, then 
% $f[Z] := \set{f(x)\mid x\in X\cap Z}$.

The extracted programs are shown in the functional programming 
language Haskell. As Haskell's syntax is very close to the usual
mathematical notation for data and functions 
%(except that in Haskell one writes e.g.\ \verb|f :: Int -> Int| 
%instead of \verb|f : Int -> Int|) 
we hope that also readers not familiar with Haskell will 
be able to understand the code. 
The Haskell code shown in this paper is self contained and can be
obtained from the author on request. 
%
%The complete Haskell programs for all examples
%in this paper can be obtained from the author's web-page 
%(\verb|http://www.cs.swan.ac.uk/~csulrich/|).

\section{Induction and coinduction}
\label{sec-ind-coind}
We briefly discuss inductive and coinductive definitions
as least and greatest fixed points of monotone set operators
and the corresponding induction and coinduction principles.
%We also have a look at the related principle of wellfounded induction. 
%
% In the discussion below the partial order 
% $(\pow{U},\tm)$ could be replaced by any complete lattice.
%
The results in this section are standard and can be found in many
logic and computer science texts. For example 
in~\cite{BuFePoSi81} inductive definitions are proof-theoretically analysed,
and in \cite{BradfieldStirling07}
least and greatest fixed points are studied in the framework
of the modal mu-calculus.

An operator $\Phi\colon\pow{U}\to\pow{U}$ (where $U$ is an arbitrary 
``universal'' set and $\pow{U}$ is the powerset of $U$) 
is \emph{monotone} if for all $X,Y\tm U$
\begin{center}
if $X\tm Y$, then $\Phi(X)\tm\Phi(Y)$ 
\end{center}
A set $X\tm U$ is \emph{$\Phi$-closed} (or a pre-fixed point of $\Phi$)
if $\Phi(X)\tm X$.
Since $\pow{U}$ is a complete lattice, $\Phi$ has a least
fixed point $\lfp{\Phi}$ (Knaster-Tarski Theorem). 
For the sake of readability we will sometimes write 
$\lfpt{X}{\Phi(X)}$ instead of $\lfp{\Phi}$.
$\lfp{\Phi}$ can be defined as the least $\Phi$-closed subset of $U$.
Hence we have the \emph{closure principle} for $\lfp{\Phi}$, 
$\Phi(\lfp{\Phi}) \tm \lfp{\Phi}$
and the \emph{induction principle} stating that for every $X\tm U$,
if $\Phi(X) \tm X$, then $\lfp{\Phi} \tm X$.
It can easily be shown that $\lfp{\Phi}$ is even a \emph{fixed point} 
of $\Phi$, i.\ e.~$\Phi(\lfp{\Phi}) = \lfp{\Phi}$ (Lambek's Lemma). 
For monotone operators $\Phi,\Psi\colon\pow{U}\to\pow{U}$ 
we define
\[\Phi\tm \Psi \dequiv \all{X\tm U}\Phi(X)\tm\Psi(X)\]
It is easy to see that the operation \emph{$\lfp{}$ is monotone}, i.\ e.\ 
if $\Phi\tm\Psi$, then $\lfp{\Phi}\tm\lfp{\Psi}$.
%
%This follows by induction on $\lfp{\Phi}$, since
%$\Phi(\lfp{\Psi})\tm\Psi(\lfp{\Psi})\tm\lfp{\Psi}$. 
%The first inclusion holds since $\Phi\tm\Psi$, the
%second inclusion holds by the closure principle for $\lfp{\Psi}$.
%
Using monotonicity of $\lfp{}$ one can easily prove, by induction, a
principle, called
\emph{strong induction}. It says that,
if $\Phi(X\cap\lfp{\Phi}) \tm X$, then $\lfp{\Phi} \tm X$.
%
%To prove this principle, set $\Psi(X) := \Phi(X\cap\lfp{\Phi})$. This defines
%an operator $\Psi\colon\pow{U}\to\pow{U}$ which clearly is monotone.
%The assumption $\Phi(X\cap\lfp{\Phi}) \tm X$ means
%$\Psi(X)\tm X$ which, by induction on $\lfp{\Psi}$, entails
%$\lfp{\Psi} \tm X$.
%Hence, it suffices to show $\lfp{\Phi}\tm\lfp{\Psi}$.
%We prove this by induction on $\lfp{\Phi}$. First note, that
%since $\Phi$ is monotone, we have $\Psi\tm\Phi$, and therefore,
%the reverse inclusion, $\lfp{\Psi}\tm\lfp{\Phi}$,
%holds, by monotonicity of $\lfp{}$.
%According to the induction principle for $\Phi$, we have to show 
%$\Phi(\lfp{\Psi})\tm\lfp{\Psi}$. Since $\lfp{\Psi}\tm\lfp{\Phi}$,
%$\Phi$ is monotone and $\lfp{\Phi}$ is $\Phi$-closed we have 
%$\Phi(\lfp{\Psi})\tm\Phi(\lfp{\Phi})\tm \lfp{\Phi}$.

Dual to inductive definitions are \emph{coinductive definitions}.
%Let again $\Phi\colon\pow{U}\to\pow{U}$ be monotone. 
A subset $X$ of $U$ is called \emph{$\Phi$-coclosed} 
(or a post-fixed point of $\Phi$) if $X\tm\Phi(X)$.
By duality, $\Phi$ has a largest fixed point $\gfp{\Phi}$ which can be
defined as the largest $\Phi$-coclosed subset of $U$. Similarly,
all other principles for induction have their coinductive counterparts.
To summarise, we have the following principles:
\begin{center}
\begin{tabular}{lcl}
\emph{Fixed point} &\hbox{}\qquad\qquad\hbox{}&
$\Phi(\lfp{\Phi})=\lfp{\Phi}$ and 
$\Phi(\gfp{\Phi})=\gfp{\Phi}$.\\[0.3em]
\emph{Monotonicity} &&
if $\Phi\tm\Psi$, then $\lfp{\Phi}\tm\lfp{\Psi}$ and 
                       $\gfp{\Phi}\tm\gfp{\Psi}$.\\[0.3em]
\emph{Induction}&&
if $\Phi(X)\tm X$, then $\lfp{\Phi}\tm X$.\\[0.3em]
\emph{Strong induction} &&
if $\Phi(X\cap\lfp{\Phi})\tm X$, then $\lfp{\Phi}\tm X$.\\[0.3em]
\emph{Coinduction}&&
if $X\tm\Phi(X)$, then $X\tm\gfp{\Phi}$.\\[0.3em]
%
%Proving $X\tm\gfp{\Phi}$ 
%by \emph{coinduction on $\gfp{\Phi}$} means
%proving $X\tm\Phi(X)$ and using the coinduction principle.
%
\emph{Strong coinduction} &&
if $X\tm\Phi(X\cup\gfp{\Phi})$, then $X\tm\gfp{\Phi}$.
\end{tabular}
\end{center}

\begin{exa}[{\bf natural numbers}]
Define $\Phi : \pow{\RR}\to\pow{\RR}$ by 
\[
\Phi(X) := 
\set{0}\cup\set{y+1\mid y\in X} 
= \set{x\mid x=0\lor \ex{y\in X}(x = y+1)} 
\] 
Then $\lfp{\Phi}=\NN=\set{0,1,2,\ldots}$.
We consider this as the \emph{definition} of the natural numbers.
%(actually, this is the definition that was given in my first-year analysis
%course; at that time I found it strange, but now it turns out that
%it is exactly right for our purposes).
%
The induction principle is logically equivalent to the usual
zero-successor-induction on $\NN$: if $X(0)$ (base) and 
$\all{x}(X(x) \to X(x+1))$ (step),
then $\all x\in\NN\,X(x)$. Strong induction weakens the step
by restricting $x$ to the natural numbers:
$\all{x\in\NN}(X(x) \to X(x+1))$.
\end{exa}

\begin{exa}[{\bf signed digits and the interval {$[-1,1]$}}]
For every signed digit $d\in\SD$ we set
$\II_d := [d/2-1/2,d/2+1/2] = \set{x\in\RR\mid |x-d/2|\le1/2}$.
Note that $\II$ is the union of the $\II_d$
and every sub interval of $\II$ of length $\le 1/2$ is contained
in some $\II_d$.
We define an operator $\out_0 : \pow{\RR}\to\pow{\RR}$ by 
\[\out_0(X) := \set{x\mid \ex{d\in\SD} (x\in\II_d \land 2x-d\in X)}\] 
and set $\coco_0 := \gfp{\out_0}$.
Since clearly $\II \tm \out_0(\II)$, it follows, by coinduction, 
that $\II\tm\coco_0$.
On the other hand $\coco_0\tm\II$, by the fixed point property. 
Hence $\coco_0 = \II$. %So, what is the point of this definition?
The point of this definition is, that the proof of ``$\II \tm \out_0(\II)$'' 
has an interesting computational content: $x\in \II$ must be given in 
such a way that it is possible to find $d\in\SD$ such that $x\in\II_d$. 
This means that $d/2$ is a \emph{first approximation} of $x$. 
The computational content of the proof of ``$\II\tm\coco_0$'', 
roughly speaking, iterates the process of finding approximations 
to $x$ ad infinitum, i.\ e.\ it computes a \emph{signed digit representation} 
of $x$ as explained in the introduction, that is, a stream $a$ of signed digits
with $\sigma(a)=x$. 
This will be made precise in Lemma~\ref{lem-approx} (Sect.~\ref{sec-coco}).
\end{exa}

\begin{exa}[{\bf lists, streams and trees}]
Let the Scott-domain $D$ be defined by the recursive domain equation
$D = 1 + D\times D$ where $1:= \set{\bot}$ is a one point domain and 
``$+$'' denotes the separated 
sum of domains (see \cite{GierzHofmannKeimelLawsonMisloveScott03} 
for information on domains).
The elements of $D$ are $\bot$ (the obligatory least element), $\Cnil := \Cleft(\bot)$, 
and $\Ccons(x,y) := \Cright(x,y)$ where $x,y\in D$.
Define $\Phi : \pow{D}\to\pow{D}\to\pow{D}$ by 
\[\Phi(X)(Y) := 
\set{\Cnil}\cup\set{\Ccons(x,y)\mid x\in X, y\in Y}\]
Clearly, $\Phi$ is monotone in both arguments. For a fixed set $X\tm D$,
$\List(X) := \lfp{(\Phi(X))}$ ($= \lfpt{Y}{\Phi(X)(Y)}$) 
can be viewed as the set of \emph{finite} 
lists of elements in $X$, and 
$\Stream(X) := \gfp{(\Phi(X))}$ ($= \gfpt{Y}{\Phi(X)(Y)}$) 
as the set of \emph{finite or infinite} lists or 
\emph{streams} of elements in $X$. 
%
%Hence, for example, 
%$\lfp{(\Phi(\set{\Cnil}))}$ is a copy of the natural numbers. 
%
Since $\lfp{}$ is monotone the operator 
$\List :\pow{D}\to\pow{D}$ is again monotone.
Hence we can define $\tree := \gfp{\List}\tm D$ which is the set
of finitely branching wellfounded or non-wellfounded trees.
On the other hand, $\tree' := \lfp{\Stream}$ consist of all finitely 
or infinitely branching wellfounded trees.
The point of this example is that the definition of $\tree$ 
is similar to the characterization of uniformly continuous 
functions from $\II^n$ to $\II$ in Sect.~\ref{sec-coco},
the similarity being the fact that it is a coinductive definition with an 
inductive definition in its body. The set $\coco_0$ of the previous example
corresponds to the case $n=0$ where the inner inductive definition 
is trivial.
\end{exa}

\noindent
{\bf Formalization}\quad
We now sketch the formal system for reasoning about
inductive and coinductive definitions (a full account is given
in~\cite{BergerCCA09,SeisenBerger10}).
Since we only consider (co)inductive definitions of subsets of a given
``universal set'' we can work in a many-sorted first-order predicate
logic with free predicate variables extended by the possibility to
form for a predicate $\preda$ which is strictly positive (s.p.) in a
predicate variable $X$ the predicates $\lfpt{X}{\preda}$ and
$\gfpt{X}{\preda}$ denoting the least and greatest fixed points of the
monotone set operator defined by $\preda$. For example, $\preda$ could
be given as a comprehension term $\set{\vec x \mid A(\vec x, X)}$
where $A(\vec x, X)$ is a formula which is s.p. in $X$. The formula
$A(\vec x,X)$ may have further free object and predicate variables.
``Nested'' inductions/coinductions such as $\gfpt{X}\lfpt{Y}\set{\vec
  x \mid A(\vec x, X, Y)}$, where $A(\vec x,X,Y)$ is strictly positive
in $X$ and $Y$, are allowed.  Hence the second example above can be
formalized.  As a proof calculus we use intuitionistic natural
deduction with axioms expressing (co)closure and (co)induction for
(co)inductively defined predicates. Further axioms describing the
mathematical structures under consideration can be freely added as
long as (we know that) they are true and do not contain
disjunctions. The latter restriction ensures that these ``ad-hoc
axioms'' have no computational content, as will be explained in
Sect.~\ref{sec-realizability}. 
Note that, for example, the formula $\all{x\in\NN}\ex{y\in\NN}(y^2 \le
x < (y+1)^2)$ \emph{does} have computational content since the
definition of the predicate $\NN$, given in the first example,
contains a disjunction. Hence, although true, this formula must not be
used as an axiom, but needs to be proven.

In the first example object variables are of sort $\RR$ 
while in the second example they are of sort $D$
(hence $\RR$ and $D$ are the ``universal sets''). 
The second example shows how data structures
which normally would be defined as initial algebras or final coalgebras
of endofunctors on the category of sets can be introduced in our system.
The domain-theoretic modelling has the further advantage that partial
objects (e.g. lists or trees with possibly undefined nodes and leaves) 
can be described and reasoned about as well, without extra effort.
%
% The reason why the definitions above takes place within one ``universal''
% domain $D$ is that in this way the (co)inductively defined sets can be
% obtained simply as least and greatest fixed points of operators on the
% partially ordered powerset of $D$. Of course, lists, streams and trees
% can be defined directly in the category of sets, as initial algebras
% and final coalgebras of certain functors, but these constructions
% would \emph{not} take place within a partial order.
%
We believe that, by restricting ourselves to 
categories which are just powersets, partially ordered by inclusion, the
constructions become easier to understand for non-category-theorists,
and the formal system sketched above is simpler than one describing
initial algebras and final coalgebras of functors in general.

\section{Program extraction from proofs}
\label{sec-realizability}
In this section we briefly explain how we extract programs from proofs.
Rather than giving technical definitions
we only sketch the formal framework and explain the extraction method 
by means of simple examples, which hopefully provide a 
good intuition also for non-experts. 
%
%and then make some general remarks
%concerning the computational content of induction and coinduction.
%
More details and full correctness proofs can be found in \cite{BergerCCA09}
and \cite{SeisenBerger10}.

The method of program extraction we are using is 
based on an extension and variation of Kreisel's 
\emph{modified realizability} \cite{Kreisel59}.
The \emph{extension} concerns the addition of inductive and coinductive
predicates. Realizability  %, more precisely $q$-realizability, 
for such predicates has been studied previously, in the
slightly different context of $\mathbf{q}$-realizability 
by Tatsuta~\cite{Tatsuta98}.
The \emph{variation} concerns the fact that we are treating the
first-order part of the language (i.\ e.\ quantification over
individuals) in a `uniform' way, that is, realizers do not depend on
the individuals quantified over.  This is similar to the common 
uniform treatment of
second-order variables~\cite{Troelstra73}.  The argument is that
an arbitrary subset of a set is such an abstract (and even vague)
entity so that one should not expect an algorithm to depend on it.
With a similar argument one may say that individuals of an abstract
mathematical structure ($\RR$, model of set-theory, etc.) are
unsuitable as inputs for programs.
Hence, a realizer of a formula $\all{x}A(x)$ is an object $a$
such that $a$ realizes $A(x)$ for \emph{all} $x$ where $a$ does not depend on $x$.
A realizer of a formula $\ex{x}A(x)$ is an object $a$ such that
$a$ realizes $A(x)$ for \emph{some} $x$. Note that the witness $x$ is not part of the
realizer $a$.
But which data should a program then depend on and which should it produce?
The answer is: data defined by the `propositional skeletons' of 
formulas and `canonical' proofs. 
%
% For example, the propositional 
% skeleton of the canonical proof that $3$ is a natural number records the 
% three-fold application of the successor clause 
% $\all{x}(\NN(x)\to\NN(x+1))$ 
% to the base clause $\NN(0)$, hence it can be viewed as the 
% unary representation of the number $3$.

\medskip

\noindent
{\bf Example (parity)}\quad
Let us extract a program from a proof of
%
%\[\all{x,y}(\NN(x) \land \NN(y) -> \NN(x+y)\]
%
\begin{equation}
\label{eq-parity}
\all{x}(\NN(x) \imp \ex{y}(x = 2y \lor x = 2y+1))
\end{equation}
where the variable $x$ ranges over real numbers and 
the predicate $\NN$ is defined as in the example in 
Sect.~\ref{sec-ind-coind}, i.\ e.
\begin{equation}
\label{eq-N}
\NN := \lfpt{X}{\set{x \mid x = 0 \lor \ex{y}(X(y) \land x = y+1)}}
\end{equation}
The type corresponding to (\ref{eq-N}) is obtained by
the following \emph{type extraction\/}: % process:
\begin{enumerate}[$\bullet$]
\item replace every atomic formula of the form $X(t)$ by a type variable $\alpha$
associated with the predicate variable $X$,
\item replace other atomic formulas by the unit or `void' type $\one$,
\item delete all quantifiers and object terms (i.\ o.\ w.\ remove all 
first-order parts),
\item replace $\lor$ by $+$ (disjoint sum) and 
      $\land$ by $\times$ (cartesian product),
\item carry out obvious simplifications (e.g.\ 
replace $\alpha\times\one$ by $\alpha$).
\end{enumerate} 
Hence we arrive at the type definition
\[\nat := \lfpt{\alpha}{\one + \alpha}\]
In Haskell we can define this type as
%(using the constructors \verb|Zero| and \verb|Succ| for $0$ and successor)
%
\begin{code}
data Nat = Zero | Succ Nat      -- data
   deriving Show
\end{code}
%
%but it is more convenient to use the built-in data type of integers
%instead
%
%\begin{code}
%type Nat = Int  -- (comment) we use only non-negative integers
%\end{code}
%
The line ``\verb|deriving Show|'' creates a default printing method
for values of type \verb|Nat|.
The comment ``\verb|-- data|'' indicates that we intend to use the 
recursive data type \verb|Nat| as an inductive data type (or initial algebra).
This means that the ```total'', or ``legal'' elements are inductively generated
from \verb|Zero| and \verb|Succ|. The natural (domain-theoretic) semantics
of \verb|Nat| also contains, for example, an ``infinite'' element 
defined recursively by \verb|infty = Succ infty| which is not total
in the inductive interpretation of \verb|Nat|. In a coinductive interpretation
(usually indicated by the comment \verb|-- codata|) \verb|infty|
would count as total.\footnote{That Haskell does not distinguish
between the inductive and the coinductive interpretation is 
justified by the limit-colimit-coincidence in the domain-theoretic
semantics~\cite{AbramskyJung94}.}

By applying type extraction to (\ref{eq-parity}) we see that a
program extracted from a proof of this formula will have type $\nat \to \one+\one$.
%(where ``$\nat\to$'' does not come from ``$\all{x}$'' 
%in~(\ref{eq-parity}), but from ``$\NN(x)\imp$'').
%
%(where the arrow is now to be read as a function space constructor).
%
By identifying the two-element type $\one+\one$ with the Booleans we get
the Haskell signature
\begin{code}
parity :: Nat -> Bool
\end{code}
The definition of \verb|parity| can be extracted from the obvious
inductive proof of (\ref{eq-parity}): For the base, $x=0$, we take $y = 0$
to get $x=2y$. In the step, $x+1$, we have, by i.\ h.\ some
$y$ with $x = 2y \lor x = 2y+1$. In the first case $x+1=2y+1$,
in the second case $x+1=2(y+1)$. The Haskell program extracted from 
this proof is 
\begin{code}
parity Zero     = True
parity (Succ x) = case parity x of {True -> False ; False -> True} 
\end{code}
%\begin{code}
%parity :: Nat -> Bool
%parity 0     = True
%parity (x+1) = case parity x of {True -> False ; False -> True} 
%\end{code}
%
%i.\ e. \verb|parity (Succ x) = not (parity x)|. 
%
If we wish to compute not only the parity, but as well the
rounded down half of $x$ (i.\ e.\ quotient and remainder),
we just need to relativize the
quantifier $\ex{y}$  in (\ref{eq-parity})   to $\NN$ 
(i.\ e.\ $\all{x}(\NN(x) \imp \ex{y}(\NN(y) \land (x = 2y \lor x = 2y+1))$)
and use in the proof the fact that $\NN$ is closed under the successor 
operation. The extracted program is then
\begin{code}
parity1 :: Nat -> (Nat,Bool)
parity1 Zero     = (Zero,True)
parity1 (Succ x) = case parity1 x of 
                    {(y,True)  -> (y,False) ; 
                     (y,False) -> (Succ y,True)} 
\end{code}
In order to try these programs out it is convenient to have a
function that transforms built-in integers into elements of \verb|Nat|.
\begin{code}
iN :: Integer -> Nat  -- defined for non-negative integers only
iN 0     = Zero
iN (n+1) = Succ (iN n)
\end{code}
Now try \verb|parity (iN 7)| and \verb|parity1 (iN 7)|. 

The examples above show that we can get meaningful computational
content despite ignoring the first-order part of a proof. Moreover, we
can fine-tune the amount of computational information we extract from
a proof by simple modifications of formulas and proofs. Note also that
we used arithmetic operations on the reals and their
arithmetic laws without implementing or proving them. Since these laws
can be written as equations (or conditional equations) their
associated type is void. This ensures that it is only their truth that
matters, allowing us to treat them as ad-hoc axioms without bothering
to derive them from basic axioms. In general, a formula containing
neither disjunctions nor free predicate variables has always a void
type and can therefore be taken as an axiom as long as it is true.

The reader might be puzzled by the fact that quantifiers are ignored 
in the program extraction process. Quantifiers are, of course,
\emph{not} ignored in the \emph{specification} of the extracted program,
i.\ e.\ in the definition of realizability. 
For example, the statement that
the program $p :=$\verb|parity| realizes (\ref{eq-parity})
is expressed by %the formula 
%(writing $\btrue$, $\bfalse$ for \verb|True|, \verb|False|)
%
\[\all{n,x}(\ire{n}{\NN(x)} \imp \ex{y}(p(n)=\verb|True| \land x = 2y \lor 
p(n)=\verb|False| \land x = 2y+1))\]
where $n$ ranges over $\nat$ (i.\ e.\ %the terms
\verb|Zero|, \verb|Succ Zero|, \verb|Succ(Succ Zero)|, \ldots) 
and $\ire{n}{\NN(x)}$ 
means that $n$ realizes $\NN(x)$ which in this case amounts to $x$ being the
value of $n$ in $\RR$. 
%(interpreting \verb|Zero| as $0$ and \verb|Succ| as the 
%successor function in $\RR$). 
The \emph{Soundness Theorem} for 
realizability states that the program extracted from a proof realizes the 
proven formula (cf.~\cite{SeisenBerger10}; see also e.g. 
\cite{Troelstra73}, \cite{Tatsuta98} for proofs of soundness 
for related notions of realizability).

\paragraph{Remark.}
%
%Before continuing with an example involving coinduction let us 
%explain the semantical nature of realizers. 
%
Although the example above 
seems to suggest that realizers are typed, it is in fact more convenient
to work with per se untyped realizers taken from a domain $D$ which is defined
by the recursive domain equation
\[ D = 1 + D + D + D\times D + [D\to D] \]
($[D \to D ]$ denotes the domain of continuous endofunctions on $D$).
It is well-known that such domain equations of ``mixed variance'' 
have effective solutions up to isomorphism 
(see e.g.\ \cite{GierzHofmannKeimelLawsonMisloveScott03}).
Type expressions $\one$, $\alpha$, $\rho+\sigma$, $\rho\times\sigma$,
$\rho\to\sigma$, $\lfpt{\alpha}{\rho}$, $\gfpt{\alpha}{\rho}$ with
suitable positivity conditions for fixed point types can naturally be
interpreted as subsets of $D$.%
\footnote{In fact, general recursive types $\fix{\alpha}{\rho}$ 
without positivity
restriction have a natural semantics in $D$ as (ranges of) finitary
projections~\cite{AmadioBruceLongo86}.}
Realizers extracted from proofs are
terms of an untyped $\lambda$-calculus with constructors and recursion
which denote elements of $D$. It can be shown that the value of a
program extracted from a proof of a formula $A$ lies in the denotation
of the type extracted from $A$~\cite{SeisenBerger10}. One can also
show that the denotational and operational semantics ``match''
(computational adequacy). This implies that extracted programs are
correct, both in a denotational and operational
sense~\cite{BergerCCA09}.  Note that in general a realizing term
denotes an element of $D$, but not an element of the mathematical
structure the proof is about. It is just by coincidence that in the
example above the closed terms of type $\nat$ denote at the same time
elements of $D$ and real numbers, and that both denotations 
are in a one-to-one correspondence. 
In the case of the predicate $\coco_0$ defined below (and even more so
for the predicates $\coco_n$ defined in Sect.~\ref{sec-coco}) there
is no such tight correspondence between objects satisfying a predicate 
and realizers of that fact.

\medskip

\begin{exa}[{\bf from Cauchy sequences to signed digit streams}]
In the second example of Sect.~\ref{sec-ind-coind} we defined the
set $\coco_0$  coinductively by
\begin{equation}
\label{eq-cocozero}
\coco_0  = 
     \gfpt{X}{\set{x\mid\ex{d}(\SD(d) \land \II_d(x) \land X(2x-d))}}
\end{equation}
Since $\SD(d)$ is shorthand for $d = -1 \lor d = 0 \lor d = 1$, 
and $\II_d(x)$ is shorthand for $|x-d/2|\le 1/2$, the corresponding type
is 
\begin{equation}
\label{eq-sds}
\mathrm{C0}  = \gfpt{\alpha}{(\one+\one+\one) \times \alpha}
\end{equation}
Identifying notationally the type $\one+\one+\one$ with $\SD$ 
\begin{code}
data SD  = N | Z | P  -- Negative, Zero, Positive
   deriving Show
\end{code}
we obtain that $\mathrm{C0}$ is the type of infinite
streams of signed digits, i.\ e.\ the largest fixed point of the type operator
\begin{code}
type J0 alpha = (SD,alpha)
\end{code}
This corresponds to the set operator $\out_0$ which
$\coco_0$ is the largest fixed point of.
Therefore we define (choosing \verb|ConsC0| as constructor name)
\begin{code}
data C0 = ConsC0 (J0 C0)    -- codata         
\end{code}
i.\ e.\ \verb|C0 = ConsC0 (SD,C0)|.
\end{exa}

We wish to extract a program that computes a signed digit representation
of $x\in\II$ from a fast rational Cauchy sequence converging to $x$
and vice versa.
Set
\begin{eqnarray*}
\QQ(x) &:=& \ex{n,m,k}(\NN(n)\land\NN(m)\land\NN(k) \land x = (n-m)/k)\\
A(x)   &:=& \all{n}(\NN(n) \imp \ex{q}(\QQ(q)\land |x-q|\le 2^{-n}))
\end{eqnarray*}
Constructively, $A(x)$ means that there is a fast Cauchy sequence of
rational numbers converging to $x$. Technically, this is expressed 
by the fact that the realizers of $A(x)$ are precisely such sequences.
On the other hand, realizers of $\coco_0(x)$ are exactly the
infinite streams of signed digits $a$ such that $\sigma(a)=x$
%($\sigma$ was defined in the Introduction). 
In general, 
realizability for inductive resp.\ coinductive predicates 
%(such as, for example,  $\NN$ and $\coco_0$) 
is defined in a straightforward way,
again as an inductive resp.\ coinductive definition 
(see~\cite{BergerCCA09,SeisenBerger10} for details).
\begin{lem}
\label{lem-approx}
\begin{equation}
\label{eq-cauchy-coco}
\all{x}(\II(x) \land A(x) \bimp \coco_0(x))
\end{equation}
\end{lem}
\begin{proof}
  To prove the implication from left to right we show $\II\cap A \tm
  \coco_0$ by coinduction, i.\ e.\ we show $\II\cap A \tm \out_0(\II\cap
  A)$.
  Assume $\II(x)$ and $A(x)$. We have to show (constructively!)
  $\out_0(\II\cap A)(x)$, i.\ e.\ we need to find $d\in\SD$ such that
  $x\in \II_d$ and $2x-d\in\II\cap A$.  Since, clearly 
  the assumption $A(x)$ implies $A(2x-d)$ for any 
  $d\in\SD$, and furthermore $x\in \II_d$ holds iff 
  $2x-d\in\II$, we only need to find some signed digit $d$
  such that $x\in\II_d$.  The assumption $A(x)$,
  used with $n=2$, yields a rational number $q$ with $|x-q|\le 1/4$. It is
  easy to find (constructively!)  a signed digit $d$ such that
  $[q-1/4,q+1/4]\cap\II \tm \II_d$.  For that $d$ we have $x\in\II_d$.

For the converse implication we show 
$\all{n}(\NN(n) \imp \all{x}(\coco_0(x) \imp 
\ex{q}(\QQ(q)\land |x-q|\le 2^{-n}))$
by induction on $\NN(n)$ using the coclosure axiom for $\coco_0$.
We leave the details as an exercise for the reader. 
\end{proof}
The type corresponding to the predicate $\QQ$ is
$\nat\times\nat\times\nat$, which we however implement by Haskell's
built-in rationals, since it is only the arithmetic operations on
rational numbers that matter, whatever the representation. (It is
possible - and instructive as an exercise - to extract implementations
of the arithmetic operations on rational numbers w.r.t.\ the
representation $\nat\times\nat\times\nat$ from proofs that $\QQ$ is
closed under these operations. In order to obtain reasonably efficient
programs one has to modify the definition of $\QQ$ by requiring $n-m$
and $k$ to be relatively prime.)
The type of the predicate $A$ is $\nat\to\rat$. 
The program extracted from the first part of the proof of
Lemma~\ref{lem-approx} is
\begin{code}
cauchy2sd :: (Nat -> Rational) -> C0
cauchy2sd = coitC0 step
\end{code}
where \verb|step| is the program extracted from the proof of
$\II\cap A \tm \out_0(\II\cap A)$:
\begin{code}
step :: (Nat -> Rational) -> J0(Nat -> Rational)
step f = (d,f')  where
  q = f (Succ (Succ Zero))
  d = if q > 1/4 then P else if abs q <= 1/4 then Z else N
  f' n = 2 * f (Succ n) - fromSD d

fromSD :: SD -> Rational
fromSD d = case d of {N -> -1 ; Z -> 0 ; P -> 1}
\end{code}
The program \verb|coitC0| is a polymorphic ``coiterator''
realizing the coinduction scheme 
$X\tm\out_0(X) \imp X\tm\gfp{\out_0}$: % (recall that $\gfp{\out_0}=\coco_0$).
%
% coit s x = ConsC0 . mapJ0 (coit s) . s   -- . = composition
\begin{code}
coitC0 :: (alpha -> J0 alpha) -> alpha -> C0
coitC0 s x = ConsC0 (mapJ0 (coitC0 s) (s x))

mapJ0 :: (alpha -> beta) -> J0 alpha -> J0 beta
mapJ0 f (d,x) = (d,f x)
\end{code}
An equivalent definition of \verb|coitC0| would be
\begin{code}
coitC0' s x = ConsC0 (d,coitC0' s y) where (d,y) = s x
\end{code}
The program extracted from the second part of the proof of
Lemma~\ref{lem-approx} is
\begin{code}
sd2cauchy :: C0 -> (Nat -> Rational)
sd2cauchy c n = aux n c  where
  aux Zero c = 0
  aux (Succ n) (ConsC0 (d,c)) = (fromSD d + aux n c)/2  
\end{code}
In order to try out the programs \verb|cauchy2sd| and \verb|sd2cauchy|
it is convenient to have translations between the types \verb|C0|
and Haskell's type of infinite streams of signed digits (below, ``\verb|:|''
is the cons operation for lists).
\begin{code}
c0s :: C0 -> [SD]
c0s (ConsC0 (d,c)) = d : c0s c

sc0 :: [SD] -> C0
sc0 (d:ds) = ConsC0 (d,sc0 ds)
\end{code}
Now evaluate 
\verb|let {f x = 2/3} in take 10 (c0s(cauchy2sd f))| and\\
\verb!let {ds = P:Z:ds} in [sd2cauchy (sc0 ds) (iN n) | n <- [0..9]]!\\ 
(\verb|ds| is the infinite list \verb|[P,Z,P,Z,...]| and
\verb![e(n) | n <- [0..9]]! is a list comprehension expression denoting
\verb|[e(0),...,e(n)]|).

We hope that the examples above give enough hints for understanding 
program extraction from coinductive proofs. 
Here is a sketch of how it works in general. Suppose $\gfp{\Phi}$
is a coinductive predicate defined by a strictly positive set operator
$\Phi$ ($\out_0$ in our example), e.g.\ $\Phi(X)=\set{\vec x\mid A(X,\vec x)}$
where $A$ is s.p. in $X$.
From $\Phi$ one extracts a s.p. type operator 
\begin{code}
data Phi alpha = PhiDef(alpha)  -- to be replaced by a suitable 
                                -- extracted type definition
\end{code}
(\verb|J0| in our example). Due to the strict positivity of \verb|Phi|
one can define, by structural recursion on the definition of
\verb|Phi alpha|, a polymorphic map operation
\begin{code}
mapPhi :: (alpha -> beta) -> Phi alpha -> Phi beta
mapPhi = undefined -- to be replaced by an extracted program
\end{code}
and from that, recursively, the coiterator 
%
%(since the definition of
%the type operator \verb|Phi| is missing the Haskell
%code in the rest of this section does \emph{not} compile)
%
\begin{code}
coitFix :: (alpha -> Phi alpha) -> alpha -> Fix
coitFix s x = ConsFix (mapPhi (coitFix s) (s x))
\end{code}
where \verb|Fix| is the largest fixed point of \verb|Phi|:
\begin{code}
data Fix = ConsFix (Phi Fix)  -- codata
\end{code}
The program extracted from a coinductive proof of $X\tm\gfp{\Phi}$
is 
%then 
\verb|coitFix step| where \verb|step :: alpha -> Phi alpha|
is the program extracted from the proof of $X\tm\Phi(X)$ (\verb|alpha|
is the type corresponding to the predicate $X$).
For inductive proofs the construction is similar: One defines recursively
an ``iterator''
\begin{code}
itFix :: (Phi alpha -> alpha) -> Fix -> alpha 
itFix s (ConsFix z) = s (mapPhi (itFix s) z)
\end{code}
where the type \verb|Fix| is now viewed as the \emph{least} fixed point of 
\verb|Phi|.
The program extracted from an inductive proof of $\lfp{\Phi}\tm X$
is now \verb|itFix step| where \verb|step :: Phi alpha -> alpha|
is extracted from the proof of $\Phi(X)\tm X$. 
It is a useful exercise to re-program the data type \verb|Nat| and the
iteratively defined functions \verb|parity|, \verb|parity1|
and \verb|sd2cauchy| following strictly this general scheme.
%
%More details on the
%realizers of induction and coinduction can be found in \cite{SeisenBerger10}.
%
The above sketched computational interpretations of induction and 
coinduction and more general recursive schemes can be  
derived from category-theoretic considerations using the initial 
algebra/final coalgebra interpretation
of least and greatest fixed points (see for example 
\cite{Malcolm90,HancockSetzer03,AbelMatthesUustalu05,CaprettaUustaluVene06}).
%
%in revised paper mention cca-jucs paper

\section{Coinductive definition of uniform continuity}
\label{sec-coco}
For every $n$ we define a set $\coco_n\tm\funn{n}$ for which
we will in Sect.~\ref{sec-digit} show that it coincides with the set of
uniformly continuous functions from $\II^n$ to $\II$.

In the following we let 
$n,m,k,l,i$ range over $\NN$, $p,q$ over $\QQ$, $x,y,z$ over $\RR$,
and $d,e$ over $\SD$. Hence, for example, $\ex{d}A(d)$ is shorthand
for $\ex{d}(\SD(d)\land A(d))$ and $\bigwedge_{d}A(d)$ abbreviates 
$A(-1) \land A(0) \land A(1)$. 
%
%We also set $\ivl{q}{l} := \set{x\in\RR\mid |x-q|\le 2^{-l}}$. 
%Hence $\II =\ivl{0}{0}$. 
%
We define average functions and their inverses
\begin{eqnarray*}
\av{d}\colon\RR\to\RR, && \av{d}(x) := \frac{x+d}{2}\\
\va{d}\colon\RR\to\RR, && \va{d}(x) := 2x-d
\end{eqnarray*}
Note that $\av{d}[\II] = \II_d$ and hence
$f[\II]\tm\II_d$ iff $(\va{d}\comp f)[\II]\tm\II$.
We also need extensions of the average functions to $n$-tuples
\[\avn{i}{d}(x_1,\ldots,x_{i-1},x_i,x_{i+1},\ldots,x_n) :=
(x_1,\ldots,x_{i-1},\av{d}(x_i),x_{i+1},\ldots,x_n)\]
We define an operator 
$\outb_n\colon\pow{\funn{n}}\to\pow{\funn{n}}\to\pow{\funn{n}}$ by
\[\outb_n(X)(Y) := \set{f \mid \ex{d}(f[\II^n]\tm\II_d \land X(\va{d}\comp f)) 
  \lor
 \ex{i}\bigwedge_{d}Y(f\comp\avn{i}{d})}\]
Since $\outb_n$ is strictly positive in both arguments, we can
define an operator $\out_n\colon\pow{\funn{n}}\to\pow{\funn{n}}$  
by 
\[\out_n(X) := \lfp{(\outb_n(X))} = \lfpt{Y}{\outb_n(X)(Y)}\]
Hence, $\out_n(X)$ is the set inductively defined by the following
two rules:
\begin{equation}
\label{eq-outn-w}
\ex{d}(f[\II^n]\tm\II_d \land X(\va{d}\comp f)) \imp \out_n(X)(f)
\end{equation}
\begin{equation}
\label{eq-outn-r}
\ex{i}\bigwedge_{d}\out_n(X)(f\comp\avn{i}{d}) \imp \out_n(X)(f)
\end{equation}
Since, as mentioned in Sect.~\ref{sec-ind-coind}, the operation $\lfp{}$ 
is monotone, $\out_n$ is monotone as well. Therefore, 
we can define $\coco_n$ as the largest fixed point of $\out_n$, 
\begin{equation}
\label{eq-cocon-def}
\coco_n = \gfp{\out_n} = \gfpt{X}{\lfpt{Y}{\outb_n(X)(Y)}} 
\end{equation}
Note that for $n=0$ the second argument $Y$ of $\outb_n$ becomes
a dummy variable, and therefore $\out_0$ and $\coco_0$ are the 
same as in the corresponding example in Sect.~\ref{sec-ind-coind}.
Note also that if $f\in\coco_n$, then $f[\II^n]\tm\II_d\tm\II$ for some
$d\in\SD$ since
$\coco_n = \lfpt{Y}{\outb_n(\coco_n)(Y)}=\outb_n(\coco_n)(\coco_n)$.
%
%It is possible to show that a function $f:\II^n\to\RR$ lies in
%$\coco_n$ iff $f$ is u.\ c.\ and $f[\II^n]\tm\II$,
%and this can be proven constructively, w.r.t. the standard 
%constructive notion of uniform continuity 
%(see e.g.\ \cite{Schwichtenberg08}).
%%
%From such a proof programs can be extracted that translate
%between the usual representation of u.\ c.\ functions and the one given
%by the predicate $\coco_n$~\cite{Berger09}. 
%%
%However, experiments have shown that such programs are not very efficient.
%Better results are extracted from proofs for specific families of
%functions which we study in the next section.

The type corresponding to the formula $\outb_n(X)(Y)$ is
$\funa_n(\tva)(\tvb) := \SD\times\alpha + \NN_n\times\beta^3$
where $\NN_n := \set{1,...,n}$.
Therefore, the type of $\out_n(X)$ is 
$\lfpt{\beta}{\SD\times\alpha + \NN_n\times\beta^3}$ 
which is the type of finite ternary trees with 
indices $i\in\NN_n$ attached to the inner nodes and pairs 
$(d,x)\in\SD\times\alpha$ attached to the leaves. 
Consequently, the type of $\coco_n$ is
\begin{equation}
\label{eq-treen-def}
\gfpt{\alpha}{\lfpt{\beta}{\SD\times\alpha + \NN_n\times\beta^3}}
\end{equation}
This is the type of non-wellfounded trees obtained by 
infinitely often stacking the finite trees on top of each other,
i.\ e.\ replacing in a finite tree each $x$ in a leaf by another 
finite tree and repeating the process in the 
substituted trees ad infinitum.
Alternatively, the elements of~(\ref{eq-treen-def}) can be described as 
non-wellfounded trees without leaves such that 
\begin{enumerate}[$-$]
\item each node is either a 
\begin{enumerate}[\quad]
\item \emph{writing node} labelled with a signed digit and with 
        one subtree, or a
\item \emph{reading node} labelled with an index $i\in\NN_n$ and with 
        three subtrees;        
\end{enumerate}
\item each path has infinitely many writing nodes.
\end{enumerate}
The interpretation of such a tree as a stream transformer is easy.
Given $n$ signed digit streams $a_1,\ldots,a_n$ as inputs, run through 
the tree and output a signed digit stream as follows:
\begin{enumerate}[\quad 1.]
\item At a writing node $(d,t)$ output $d$ and continue with the 
        subtree $t$.
\item At a reading node $(i,(t_d)_{d\in\SD})$ continue with
        $t_d$, where $d$ is the head of $a_i$, and replace $a_i$ by its tail.
\end{enumerate}
Fig.~\ref{fig-tree} shows an initial segment of a tree representing
the function 
\[f\colon\II\to\II,\quad f(x) = \frac{2}{3}(1-x^2)-1\] 
which is an instance of the family of logistic
maps discussed in Sect.~\ref{sec-digit}.
\setlength{\unitlength}{0.6em}
\begin{figure}
\begin{tikzpicture}[inner sep = 0.4mm,fill=blue!20]
\path(3.0,0.0) node(a) [circle,draw] {N}
(3.0,-1.0) node(aa) [circle,draw] {}
(-1.0,-2.0) node(aaa) [circle,draw] {}
(-2.33,-3.0) node(aaaa) [circle,draw] {N}
(-2.33,-4.0) node(aaaaa) [circle,draw] {}
(-2.77,-5.0) node(aaaaaa) [circle,draw] {}
(-2.33,-5.0) node(baaaaa) [circle,draw] {}
(-1.89,-5.0) node(caaaaa) [circle,draw] {P}
(-1.0,-3.0) node(baaa) [circle,draw] {Z}
(-1.0,-4.0) node(abaaa) [circle,draw] {}
(-1.44,-5.0) node(aabaaa) [circle,draw] {N}
(-1.0,-5.0) node(babaaa) [circle,draw] {Z}
(-0.56,-5.0) node(cabaaa) [circle,draw] {Z}
(0.33,-3.0) node(caaa) [circle,draw] {Z}
(0.33,-4.0) node(acaaa) [circle,draw] {P}
(0.33,-5.0) node(aacaaa) [circle,draw] {}
(-0.11,-6.0) node(aaacaaa) [circle,draw] {N}
(0.33,-6.0) node(baacaaa) [circle,draw] {Z}
(0.77,-6.0) node(caacaaa) [circle,draw] {Z}
(3.0,-2.0) node(baa) [circle,draw] {Z}
(3.0,-3.0) node(abaa) [circle,draw] {P}
(3.0,-4.0) node(aabaa) [circle,draw] {}
(1.67,-5.0) node(aaabaa) [circle,draw] {}
(1.23,-6.0) node(aaaabaa) [circle,draw] {N}
(1.67,-6.0) node(baaabaa) [circle,draw] {Z}
(2.11,-6.0) node(caaabaa) [circle,draw] {Z}
(3.0,-5.0) node(baabaa) [circle,draw] {Z}
(3.0,-6.0) node(abaabaa) [circle,draw] {P}
(3.0,-7.0) node(aabaabaa) [circle,draw] {}
(2.56,-8.0) node(aaabaabaa) [circle,draw] {}
(3.0,-8.0) node(baabaabaa) [circle,draw] {Z}
(3.44,-8.0) node(caabaabaa) [circle,draw] {}
(4.33,-5.0) node(caabaa) [circle,draw] {}
(3.89,-6.0) node(acaabaa) [circle,draw] {Z}
(4.33,-6.0) node(bcaabaa) [circle,draw] {Z}
(4.77,-6.0) node(ccaabaa) [circle,draw] {N}
(7.0,-2.0) node(caa) [circle,draw] {}
(5.67,-3.0) node(acaa) [circle,draw] {Z}
(5.67,-4.0) node(aacaa) [circle,draw] {P}
(5.67,-5.0) node(aaacaa) [circle,draw] {}
(5.23,-6.0) node(aaaacaa) [circle,draw] {Z}
(5.67,-6.0) node(baaacaa) [circle,draw] {Z}
(6.11,-6.0) node(caaacaa) [circle,draw] {N}
(7.0,-3.0) node(bcaa) [circle,draw] {Z}
(7.0,-4.0) node(abcaa) [circle,draw] {}
(6.56,-5.0) node(aabcaa) [circle,draw] {Z}
(7.0,-5.0) node(babcaa) [circle,draw] {Z}
(7.44,-5.0) node(cabcaa) [circle,draw] {N}
(8.33,-3.0) node(ccaa) [circle,draw] {N}
(8.33,-4.0) node(accaa) [circle,draw] {}
(7.89,-5.0) node(aaccaa) [circle,draw] {P}
(8.33,-5.0) node(baccaa) [circle,draw] {}
(8.77,-5.0) node(caccaa) [circle,draw] {}
;
\draw[thick] (a) -- (aa);
\draw[thick] (aa) -- (aaa);
\draw[thick] (aa) -- (baa);
\draw[thick] (aa) -- (caa);
\draw[thick] (aaa) -- (aaaa);
\draw[thick] (aaa) -- (baaa);
\draw[thick] (aaa) -- (caaa);
\draw[thick] (aaaa) -- (aaaaa);
\draw[thick] (aaaaa) -- (aaaaaa);
\draw[thick] (aaaaa) -- (baaaaa);
\draw[thick] (aaaaa) -- (caaaaa);
\draw[thick] (baaa) -- (abaaa);
\draw[thick] (abaaa) -- (aabaaa);
\draw[thick] (abaaa) -- (babaaa);
\draw[thick] (abaaa) -- (cabaaa);
\draw[thick] (caaa) -- (acaaa);
\draw[thick] (acaaa) -- (aacaaa);
\draw[thick] (aacaaa) -- (aaacaaa);
\draw[thick] (aacaaa) -- (baacaaa);
\draw[thick] (aacaaa) -- (caacaaa);
\draw[thick] (baa) -- (abaa);
\draw[thick] (abaa) -- (aabaa);
\draw[thick] (aabaa) -- (aaabaa);
\draw[thick] (aabaa) -- (baabaa);
\draw[thick] (aabaa) -- (caabaa);
\draw[thick] (aaabaa) -- (aaaabaa);
\draw[thick] (aaabaa) -- (baaabaa);
\draw[thick] (aaabaa) -- (caaabaa);
\draw[thick] (baabaa) -- (abaabaa);
\draw[thick] (abaabaa) -- (aabaabaa);
\draw[thick] (aabaabaa) -- (aaabaabaa);
\draw[thick] (aabaabaa) -- (baabaabaa);
\draw[thick] (aabaabaa) -- (caabaabaa);
\draw[thick] (caabaa) -- (acaabaa);
\draw[thick] (caabaa) -- (bcaabaa);
\draw[thick] (caabaa) -- (ccaabaa);
\draw[thick] (caa) -- (acaa);
\draw[thick] (caa) -- (bcaa);
\draw[thick] (caa) -- (ccaa);
\draw[thick] (acaa) -- (aacaa);
\draw[thick] (aacaa) -- (aaacaa);
\draw[thick] (aaacaa) -- (aaaacaa);
\draw[thick] (aaacaa) -- (baaacaa);
\draw[thick] (aaacaa) -- (caaacaa);
\draw[thick] (bcaa) -- (abcaa);
\draw[thick] (abcaa) -- (aabcaa);
\draw[thick] (abcaa) -- (babcaa);
\draw[thick] (abcaa) -- (cabcaa);
\draw[thick] (ccaa) -- (accaa);
\draw[thick] (accaa) -- (aaccaa);
\draw[thick] (accaa) -- (baccaa);
\draw[thick] (accaa) -- (caccaa);
\end{tikzpicture}
\caption{An initial segment of the tree of $f(x) = \frac{2}{3}(1-x^2)-1$.}
%\caption{An initial segment of the tree corresponding to the logistic map, 
%$f_a(x) = a(1-x^2)-1$, with $a = 2/3$.}
\label{fig-tree}
\end{figure}
In order to ``run'' this tree with an input stream of signed digits,
we follow the path determined by the input digits.  \verb|N|, \verb|Z|
or \verb|P| in the input stream means: go at a branching point left,
middle or right. The digits met on this path form the output
stream. For example, the input stream \verb|Z:Z:Z:Z:...| (representing
the number $0$) leads us along the spine of the tree and results in
the output stream \verb|N:Z;P:Z:P:Z:...| (representing
$\frac{-1}{2}+\frac{1}{8}+\frac{1}{32}+\ldots =
\frac{-1}{2}+\frac{1}{6}= -\frac{1}{3} = f(0)$) while the input stream
\verb|P:Z:Z:Z:...| (representing the number $\frac{1}{2}$) results in
the output stream \verb|N:Z;Z:Z:...| (representing $-\frac{1}{2}=
f(\frac{1}{2})$).
\footnote{The \LaTeX code for the display of the tree was generated 
automatically from a term denoting this tree which in turn was extracted from a 
formal proof that the function $f$ lies in $\coco_1$.} 
%
%The program 
%transforming the term into latex code was written by hand, but ultimately
%also this program should and could be extracted from a proof.

The above informally described interpretation of the elements of
$\coco_n$ as stream transformers is the extracted program of a special
case of Proposition~\ref{prop-comp-n} below which shows that the
predicates $\coco_n$ are closed under composition.
The following lemma is needed in its proof.
%
% Of course, one could prove closure under composition from
% the coincidence of $\coco_n$ with the u.\ c. functions
% and the fact that the latter are closed under composition,
% which is very easy to prove. However, the computational content
% of this proof involves composition of functions on the rationals
% which, as we will see in the case study on iterations of the logistic map,
% is computationally problematic, since the rational numbers appearing
% may become very big in terms of bit size (see \cite{Heckmann98}). 
% It is computationally better to
% prove closure under composition for $\coco_n$ directly, using induction
% and coinduction.
%
\begin{lem}
\label{lem-coco}
If $\coco_n(f)$, then $\coco_n(f\comp\avn{i}{d})$.
\end{lem}
\begin{proof}
We fix $i\in\set{1,\ldots,n}$ and $d\in\SD$ and set
\[D := \set{ f\comp\avn{i}{d} \mid \coco_n(f)}\]
We show $D\tm\coco_n$ by strong coinduction, 
i.\ e.\ we show $D\tm\out_n(D\cup\coco_n)$, i.\ e.\ 
$\coco_n\tm E$ where 
\[E:= \set{f \mid \out_n(D\cup\coco_n)(f\comp\avn{i}{d})}\]
Since $\coco_n=\out_n(\coco_n)$ it suffices to show $\out_n(\coco_n)\tm E$.
We prove this by strong induction on $\out_n(\coco_n)$, i.\ e.\ we show 
$\outb_n(\coco_n)(E\cap\out_n(\coco_n))\tm E$. 
Induction base: Assume $f[\II^n]\tm\II_{d'}$ and $\coco_n(\va{d'}\comp f)$. 
We need to show $E(f)$, i.\ e.\ $\out_n(D\cup\coco_n)(f\comp\avn{i}{d})$. 
By (\ref{eq-outn-w}) it suffices 
to show $(f\comp\avn{i}{d})[\II^n]\tm\II_{d'}$ and 
$(D\cup\coco_n)(\va{d'}\comp f\comp\avn{i}{d})$.
We have $(f\comp\avn{i}{d})[\II^n] = f[\avn{i}{d}[\II^n]] \tm f[\II^n] 
\tm\II_{d'}$. Furthermore, $D(\va{d'}\comp f\comp\avn{i}{d})$ 
holds by the assumption
$\coco_n(\va{d'}\comp f)$ and the definition of $D$.
Induction step: Assume, as strong induction hypothesis, 
$\bigwedge_{d'}(E\cap\out_n(\coco_n))(f\comp\avn{i'}{d'})$. 
We have to show $E(f)$, i.\ e.\ $\out_n(D\cup\coco_n)(f\comp\avn{i}{d})$.
If $i'=i$, then the strong induction hypothesis implies
$\out_n(\coco_n)(f\comp\avn{i}{d})$ which, by the monotonicity of
$\out_n$, in turn implies $\out_n(D\cup\coco_n)(f\comp\avn{i}{d})$.
If $i'\neq i$, then 
$\bigwedge_{d'}\avn{i'}{d'}\comp\avn{i}{d} = \avn{i}{d}\comp\avn{i'}{d'}$
and therefore, since the strong induction hypothesis implies 
$\bigwedge_{d'}E(f\comp\avn{i'}{d'})$, we have 
$\bigwedge_{d'}\out_n(D\cup\coco_n)(f\comp\avn{i}{d}\comp\avn{i'}{d'})$. 
By (\ref{eq-outn-r}) this implies
$\out_n(D\cup\coco_n)(f\comp\avn{i}{d})$.
\end{proof}

\begin{prop}
\label{prop-comp-n}
Consider $f\colon\II^n\to\RR$ and $g_i\colon\II^m\to\RR$, for $i=1,\ldots,n$.
If $\coco_n(f)$ and $\coco_m(g_1), \ldots, \coco_m(g_n)$, then 
$\coco_m(f\comp(g_1,\ldots,g_n))$.
\end{prop}
\begin{proof}
We prove the proposition by coinduction, i.\ e. we set
\[D := \set{f\comp(g_1,\ldots,g_n) \mid  
\coco_n(f),\ \coco_m(g_1),\ \ldots,\ \coco_m(g_n)}\]
and show that $D\tm\out_m(D)$, i.\ e. $\coco_n\tm E$ where
\[E := \set{f\in\funn{n}\mid\all{\vec g}(\coco_m(\vec g)\imp
\out_m(D)(f\comp\vec g))}\]
and $\coco_m(\vec g) := \coco_m(g_1)\land \ldots \land \coco_m(g_n)$. 
Since $\coco_n=\out_n(\coco_n)$ it suffices to show
$\out_n(\coco_n)\tm E$.
We do an induction on $\out_n(\coco_n)$,
i.\ e.\ we show $\outb_n(\coco_n)(E)\tm E$.
Induction base: Assume $f[\II^n]\tm\II_d$, $\coco_n(\va{d}\comp f)$ 
and $\coco_m(\vec g)$.
We have to show $\out_m(D)(f\comp\vec g))$. By (\ref{eq-outn-w})
it suffices to show 
$(f\comp \vec g)[\II^m]\tm\II_d$
and
$D(\va{d}\comp f\comp\vec g)$. 
The first statement holds since $\vec g[\II^m]\tm\II$, the second
holds by the definition of $D$ and the assumption.
Induction step: Assume, as induction hypothesis, 
$\bigwedge_{d}E(f\comp\avn{i}{d})$. We have to show $E(f)$ , i.\ e. 
$\coco_m\tm F$ where
\[F := \{g\in\funn{m}\mid\all{\vec g}(g=g_i\land
\coco_m(\vec g)\imp\out_m(D)(f\comp\vec g))\}\]
Since $\coco_m\tm\out_m(\coco_m)$ it suffices to show
$\out_m(\coco_m)\tm F$ which we do by a side induction on $\out_m$,
i.\ e.\ we show $\outb_m(\coco_m)(F)\tm F$.
Side induction base: 
Assume $g[\II^m]\tm\II_d$ and 
$\coco_m(\va{d}\comp g)$ and $\coco_m(\vec g)$
where $g=g_i$. We have to show
$\out_m(D)(f\comp\vec g)$.
Let $\vec g'$ be obtained from $\vec g$ by replacing
$g_i$ with $\va{d}\comp g$. Since $\coco_m(\vec g')$, we have 
$\out_m(D)(f \comp\, \avn{i}{d}\comp\vec g')$, 
by the main induction hypothesis. 
But $\avn{i}{d}\comp\vec g' = \vec g$.
Side induction step: Assume $\bigwedge_{d}F(g\comp\avn{j}{d})$ 
(side induction hypothesis).
We have to show $F(g)$. Assume $\coco_m(\vec g)$
where $g=g_i$. We have to show $\out_m(D)(f\comp\vec g)$.
By (\ref{eq-outn-r}) it suffices to show 
$\out_m(D)(f\comp\vec g\comp\avn{j}{d})$ for all $d$.
Since the $i$-th element of $\vec g\comp\avn{j}{d}$ is
$g\comp\avn{j}{d}$ and, by Lemma~\ref{lem-coco},
$\coco_m(\vec g\comp\avn{j}{d})$, we can apply the 
side induction hypothesis.
\end{proof}
The program extracted from Prop.~\ref{prop-comp-n} 
composes trees.
The cases $m=0$ and $n=1$ are of particular interest.
If $m=0$, then the program interprets 
a tree in $\coco_n$ as an $n$-ary stream transformer.
In the proof the functions $\vec g$ are then just real numbers, and 
composition, $f \comp \vec g$, becomes function application, $f(\vec g)$.
Furthermore, the side induction step disappears.
If $n=1$, then the vectors $\vec g$ consist of only one function $g$
and $F$ simplifies to
$\{g\in\funn{m}\mid \out_m(D)(f\comp g)\}$.
Furthermore, the side induction step does not need Lemma~\ref{lem-coco}
anymore and becomes almost trivial.
We show the programs for the cases $n=1, m=0$ and $n=m=1$.
We use the following auxiliary programs extracted from a proof of the formula
$(X(-1)\land X(0)\land X(1))\bimp \all{d}X(d)$.
% and
%$X\tm Y \imp (X(-1)\land X(0)\land X(1)) \imp (Y(-1)\land Y(0)\land Y(1))$.
%
\begin{code}
type Triple alpha = (alpha,alpha,alpha)

appTriple :: Triple alpha -> SD -> alpha
appTriple (xN,xZ,xP) d = case d of {N -> xN ; Z -> xZ ; P -> xP}

abstTriple :: (SD -> alpha) -> Triple alpha
abstTriple f = (f N, f Z, f P)
\end{code}
%mapTriple :: (alpha -> beta) -> Triple alpha -> Triple beta
%mapTriple f (xN,xZ,xP) = (f xN,f xZ,f xP)
%
The data types associated with the operators $\outb_1$, $\out_1$ 
and the predicate $\coco_1$ as well as their associated map functions 
and (co)iterators are
\begin{code}
data K1 alpha beta = W1 SD alpha | R1 (Triple beta)
                       
mapK1 :: (alpha -> alpha') -> (beta -> beta') -> 
         K1 alpha beta -> K1 alpha' beta'
mapK1 f g (W1 d a) = W1 d (f a)
mapK1 f g (R1 (bN,bZ,bP)) = R1 (g bN,g bZ,g bP)

data J1 alpha = ConsJ1 (K1 alpha (J1 alpha))   -- data
                       
itJ1 :: (K1 alpha beta -> beta) -> J1 alpha -> beta 
itJ1 s (ConsJ1 z) = s (mapK1 id (itJ1 s) z)

mapJ1 :: (alpha -> alpha') -> J1 alpha -> J1 alpha'
mapJ1 f (ConsJ1 x) = ConsJ1 (mapK1 f (mapJ1 f) x)

data C1 = ConsC1 (J1 C1)      -- codata
 
coitC1 :: (alpha -> J1 alpha) -> alpha -> C1 
coitC1 s x = ConsC1 (mapJ1 (coitC1 s) (s x))
\end{code}
Now, the extracted programs of Proposition~\ref{prop-comp-n}.
Case $n=1, m=0$:
\begin{code}
appC :: C1 -> C0 -> C0
appC c ds = coitC0 costep (c,ds)  where 

   costep :: (C1,C0) -> J0 (C1,C0)
   costep (ConsC1 x,ds) = aux x ds  

   aux :: J1 C1 -> C0 -> J0 (C1,C0) 
   aux = itJ1 step

   step :: K1 C1 (C0 -> J0 (C1,C0)) -> C0 -> J0 (C1,C0)
   step (W1 d c') ds                = (d,(c',ds))
   step (R1 es)   (ConsC0 (d0,ds')) = appTriple es d0 ds'
\end{code}
Case $n=m=1$:
\begin{code}
compC1 :: C1 -> C1 -> C1
compC1 c1 c2 = coitC1 costep (c1,c2)  where 

   costep :: (C1,C1) -> J1 (C1,C1)
   costep (ConsC1 x1,c2) = aux x1 c2  

   aux :: J1 C1 -> C1 -> J1(C1,C1) 
   aux = itJ1 step

   step :: K1 C1 (C1 -> J1 (C1,C1)) -> C1 -> J1 (C1,C1)
   step (W1 d1 c1') c2          = ConsJ1 (W1 d1 (c1',c2))
   step (R1 es)     (ConsC1 x2) = subaux x2  where

        subaux :: J1 C1 -> J1 (C1,C1)
        subaux = itJ1 substep

        substep :: K1 C1 (J1 (C1,C1)) -> J1 (C1,C1)
        substep (W1 d2 c2') = appTriple es d2 c2'
        substep (R1 fs)     = ConsJ1 (R1 fs)
\end{code}
\paragraph{Remark.}
The cases shown above are also treated in \cite{Ghanietal09}
(without application to exact real number computation).
Whereas in \cite{Ghanietal09} the program was `guessed' and then verified,
we are able to extract the program from a proof making verification unnecessary.
Of course, one could reduce Proposition~\ref{prop-comp-n} to the
case $m=n=1$, by coding $n$ streams of single digits into one 
stream of $n$-tuples of digits.
But this would lead to less efficient programs, since it would mean
that in each reading step \emph{all} inputs are read, even those 
that might not be needed (for example, the function 
$f(x,y) = x/2 + y/100$ certainly should read $x$ more often than $y$).

%For the extracted programs, which are not shown here due 
%to lack of space, we refer the reader to~\cite{Berger09} (or
%\cite{Ghanietal09} for the case $m=n=1$).

\paragraph{Remark.}
Note that the realizability relation connecting real functions satisfying 
$\coco_1$ and trees in the type \verb|C1| is much less tight than it was 
in the case of natural numbers (where realizability provided 
a one-to-one correspondence between real numbers satisfying the 
predicate $\NN$ and elements of the type \verb|Nat|). Although, by 
coincidence, every element of the type \verb|C1| defines, via the 
program \verb|appC| a stream transformer, this stream transformer will in 
general not correspond to a real function, i.\ e.\ it will not necessarily
respect equality of reals represented by signed digit streams. The latter 
is the case only if the tree happens to realize a function $f$ 
(which is of course the case if the tree was extracted from a proof 
of $\coco_1(f)$). Moreover, a tree can realize $\coco_1(f)$ for 
different $f$ because the predicate $\coco_1$ says nothing about the 
behaviour of functions outside the interval $\II$.

In order to try out the programs \verb|appC| and \verb|compC1| one needs
examples of elements of the type \verb|C1|. Such examples will be
provided in the next section.

\section{Wellfounded induction and digital systems}
\label{sec-digit}
Now we study the principle of induction along a wellfounded relation
from the perspective of program extraction. As an important application
%of wellfounded induction 
we show that certain families of real functions
which we call digital systems are contained in $\coco_n$.
This provides a convenient tool for proving that certain functions,
for example polynomials and, more generally, uniformly continuous
functions on $\II^n$ are in $\coco_n$, and in turn allows us to extract 
implementations for these functions.

\medskip

\noindent
{\bf Wellfounded induction}\quad
Let $U$ be a set, $A$ a subset of $U$ and $<$ a binary relation on $U$. 
Define a monotone operator $\Phi : \pow{U}\to\pow{U}$ 
(depending on $A$ and $<$) by 
\[\Phi(X) := \set{x\mid\all y\in A(y<x \imp y\in X)}\]
The relation $<$ is called \emph{wellfounded} on $A$, $\Wf_A(<)$,  
if $A\tm\lfp{\Phi}$. A set $X\tm U$ is called \emph{$<$-progressive} 
on $A$, $\Prog_A(<,X)$, if $\Phi(X)\cap A\tm X$.
The principle of \emph{wellfounded induction (on $A$ along $<$ at $X$)\/},
$\WfInd_A(<,X)$, is
\[\Prog_A(<,X) \imp A\tm X \]
For the purpose of program extraction let us assume that the partial 
order $x<y$ is defined without using disjunctions and hence has no 
computational content. For example, the definition could be 
an equation $t(x,y)=0$ for some term $t(x,y)$.
The following program realizes wellfounded induction for provably
wellfounded relations (\verb|alpha| and \verb|beta| are the 
types of realizers of $A$ and $X$, respectively):
\begin{code}
wfrec :: ((alpha -> beta) -> alpha -> beta) -> alpha -> beta
wfrec prog = h  where  h  = prog h
\end{code}
\begin{prop}
\label{prop-wf}
If $\Wf_A(<)$ is provable, then \verb|wfrec| realizes $\WfInd_A(<,X)$.
\end{prop}
The proof of Prop.~\ref{prop-wf} is beyond the scope of this introductory
paper and will be given in a subsequent publication.
\paragraph{Remark.} 
One can easily prove $\Wf_A(<) \imp \WfInd_A(<,X)$
from the induction principle for $\lfp{\Phi}$ and extract
a program computing a realizer of 
$\WfInd_A(<,X)$ from a realizer of $\Wf_A(<)$. The point is, that our 
realizer of $\WfInd_A(<,X)$ does not depend on a realizer of $\Wf_A(<)$.
\paragraph{Remark.} In~\cite{Schwichtenberg08a} a Dialectica Interpretation of
a different form of wellfounded induction is given. There, the realizing 
program refers to a decision procedure for the given wellfounded relation. 

\medskip

\noindent
{\bf Digital systems}\quad
Let $(A,<)$ be a provably wellfounded relation.
A \emph{digital system} is a family 
$\FFF = (f_x : \II^n\to\II)_{x\in A}$
such that for all $x\in A$
%
%\footnote{contrary to earlier conventions we use in in the following the 
%letters $\alpha,\beta$ for elements of an arbitrary wellfounded set and 
%the letters \texttt{a}, \texttt{b}, \texttt{c} to denote type variables.} 
%
%
\[      \ex{d}(f_x[\II^n]\tm\II_d\land 
        \ex{y\in A} f_y = \va{d}\comp f_x)
  \lor  
        \ex{i}\bigwedge_{d}\ex{y\in A} (y<x \land f_y = f_x\comp\avn{i}{d})
\]
When convenient we identify the family $\FFF$ with the set 
$\set{f_x \mid x\in A}$.
\paragraph{Remark.} The definition of a digital system makes reference
to the (undecidable) equality relation between real functions. This is 
not a problem because, as explained in Section~\ref{sec-ind-coind}, 
it is not necessary for the mathematical objects and predicates to be
constructively given. It is enough to be able to formulate the necessary
axioms without using disjunctions (which is the case for the usual 
axioms for equality between functions).
\begin{prop}
\label{prop-pot}
If $\FFF$ is a digital system, then $\FFF\tm\coco_n$.
\end{prop}
\begin{proof}
Let $\FFF$ be a digital system.
We show $\FFF\tm\coco_n$ by coinduction. Hence,  
we have to show $\out_n(\FFF)(f_x)$ for all $x\in A$. 
But, looking at the definition of $\out_n(\FFF)$ and the 
properties of a digital system, this follows immediately by 
wellfounded $<$-induction on $x$. 
\end{proof}
We can extract a program from the proof of Prop.~\ref{prop-pot} that
transforms a (realization of) a digital system into a family of trees
realizing its members (case $n=1$):
\begin{code}
digitsys1 :: (alpha -> Either (SD,alpha) (Triple alpha)) 
             -> alpha -> C1
digitsys1 s = coitC1 (wfrec prog)  where

-- prog :: (alpha -> J1 alpha) -> alpha -> J1 alpha
   prog ih x =
     case s x of
      {Left (d,a)      -> ConsJ1 (W1 d a) ;
       Right (aN,aZ,aP) -> ConsJ1 (R1 (ih aN, ih aZ, ih aP))}
\end{code}

\begin{exa}[{\bf linear affine functions}]
For $\vec u,v \in \QQ^{n+1}$ define 
$\la{\vec u,v}\colon \II^n\to\RR$ by 
\[\la{\vec u,v}(\vec x) := u_1 x_1 + \ldots + u_n x_n + v\]
Clearly, $f_{\vec u,v}[\II^n] = [v-|\vec u|,v+|\vec u|]$ 
where $|\vec u|  := |u_1|+\ldots+|u_n|$.
Hence $f_{\vec u,v}[\II^n]\tm\II$ iff $|\vec u|+|v|\le 1$,
and if $|\vec u|\le 1/4$, then $f_{\vec u,v}[\II^n]\tm\II_d$ for some $d$.
Furthermore, $f_{\vec u,v}\comp\avn{i}{d} = f_{\vec u',v'}$ where
$\vec u'$ is like $\vec u$ except that the $i$-th component is halved and
$v' = v + u_id/2$. 
Hence, if $i$ was chosen such that $|u_i|\ge |\vec u|/n$,
then $|\vec u'| \le q |\vec u|$ where $q := 1-1/(2n) <1$.
Therefore, we set $A := \set{\vec u,v\in\QQ^{n+1} \mid |\vec u|+|v|\le 1}$
and define a wellfounded relation $<$ on $A$ by
\[\vec u',v' < \vec u,v \quad :\bimp\quad 
  |\vec u| \ge 1/4 \land |\vec u'| \le q |\vec u| \]
From the above it follows that
$\poly{1,n} := (\la{\vec u,v})_{\vec u,v\in A}$
is a digital system.                                                         
Hence $\poly{1,n}\tm\coco_n$, by                                             
Proposition~\ref{prop-pot}. Program extraction gives us a                    
program that assigns to each tuple of rationals                              
$\vec u,w\in A$ a tree representation of $\la{\vec u,w}$.                   
Here is the program for the case $n=1$:                                      
\begin{code}
type Rat2 = (Rational,Rational)                                            
                                                                             
linC1 :: Rat2 -> C1                                                        
linC1 = digitsys1 s  where                                                 
                                                                             
  s :: Rat2 -> Either (SD,Rat2) (Triple Rat2)                              
  s (u,v) = if abs u <= 1/4
            then let e = if v < -(1/4) then N else
                         if v > 1/4    then P 
                                       else Z
                 in Left (e,(2*u,2*v-fromSD e))                            
            else Right (abstTriple (\d -> (u/2,u*fromSD d/2+v)))           
\end{code}                                                               
In order to try this program out we introduce a utility function
that applies a function $f : \mathrm{C0}\to\mathrm{C0}$ to the
signed digit representation of a rational number $q$ and computes the
result with precision $2^{-n}$ as a rational number. 
\begin{code}
runC :: (C0 -> C0) -> Rational -> Integer -> Rational
runC f q n = sd2cauchy (f (cauchy2sd (const q))) (iN n)
\end{code}
Now we can compute, for example, the tree representation of the             
function $f(x) = \frac{1}{4}x +\frac{1}{5}$                                  
at the signed digit representation of the point $x=\frac{1}{3}$ 
with an accuracy  of $2^{-10}$ by defining 
\begin{code}
f :: C0 -> C0
f = appC (linC1 (1/4,1/5))
\end{code}
and evaluating the expression     
\verb| runC f (1/3) 10 |.
The computed result, $\frac{145}{512}$, differs from the exact result, 
$\frac{1}{4}x +\frac{1}{5} = \frac{17}{60}$, by $\frac{1}{7680} < 2^{-10}$, 
as required.
\end{exa}
                                                                             
\begin{rem} In \cite{Konecny04} it is shown                          
that the linear affine transformations are exactly the functions             
that can be represented by a finite automaton.                               
The trees computed by our program generate these automata,                   
simply because for the computation of the tree for                           
$\la{\vec u,v}$ only finitely many other indices                             
$\vec u',v'$ are used, and Haskell will construct the tree                   
by connecting these indices by pointers.                                     
\end{rem}                                                                             
\begin{exa}[{\bf iterated logistic map}]
With a similar proof as for the linear affine maps one                       
can show that all polynomials of degree 2 with rational                      
coefficients mapping $\II$ to $\II$ are in $\coco_1$.
The following program can be extracted. It takes three
rational numbers $u,v,w$ and computes a tree representation
of the function $f_{u,v,w}(x) := ux^2+vx + w$, provided
$f_{u,v,w}$ maps $\II$ to $\II$.
The programs \verb|quadWrite| and \verb|quadRead| 
compute the coefficients of the quadratic functions
$\va{e} \comp f_{u,v,w}$ and $f_{u,v,w} \comp\av{d}$ while
\verb|quadTest| tests whether $f_{u,v,w}[\II]\tm\II_d$.
Since a quadratic function may or may not have an extremal point in the 
interval $\II$ this test is more complicated than in the linear affine case.
\begin{code}
type Rat3 = (Rational,Rational,Rational)

quadC1 :: Rat3 -> C1
quadC1 = digitsys1 s  where

 s :: Rat3 -> Either (SD,Rat3) (Triple Rat3)
 s uvw = case (filter (quadTest uvw) [N,Z,P]) of
                (e:_)  -> Left (e,quadWrite uvw e)
                []     -> Right (abstTriple (quadRead uvw))

 quadWrite :: Rat3 -> SD -> Rat3
 quadWrite (u,v,w) e = (2*u , 2*v , 2*w - e')  
   where e' = fromSD e

 quadRead :: Rat3 -> SD -> Rat3
 quadRead (u,v,w) d = (u/4 , (u*d'+v)/2 , u*d'^2/4 + v*d'/2 + w)  
   where d' = fromSD d

 quadTest :: Rat3 -> SD -> Bool
 quadTest (u,v,w) e = (e'-1)/2 <= low && high <= (e'+1)/2
   where
    e'   = fromSD e  
    low  = minimum crit            -- min (f_uvw I)
    high = maximum crit            -- max (f_uvw I)
    crit = [ u+v+w, u-v+w] ++      -- [f_uvw 1, f_uvw (-1)] 
           (if u == 0 then []
            else let x = -v/(2*u)          -- extremal point
                 in if -1 <= x && x <= 1
                    then [u*x^2 + v*x + w] -- f_uvw x
                    else [])       
\end{code}
In particular the so-called logistic map (transformed to $\II$),             
defined by 
\[f_a(x) = a(1 - x^2) - 1,\]
is in $\coco_1$ for each rational number $a\in[0,2]$.                        
\begin{code}
lmaC1 :: Rational -> C1
lmaC1 a = quadC1 (-a,0,a-1)
\end{code}
Exact computation of iterations of the logistic map 
%
%on $[0,1]$ 
%
were studied in~\cite{Blanck05} and \cite{Plume98}. 
In order to test the performance of our implementation
with these maps we use a generalized exponentiation function that raises
a value $x$ to the power $n$ ($>0$) with respect to 
an arbitrary binary function $g$ 
as ``multiplication'':
\begin{code}
gexp :: (alpha -> alpha -> alpha) -> alpha -> Int -> alpha
gexp g x 1 = x
gexp g x (n+1) = g (gexp g x n) x
\end{code}                                   
Now we define a tree representing the 100-fold iteration of the logistic
map $f_2$
\begin{code}
t100 :: C1
t100 = gexp compC1 t1 100  where  t1 = lmaC1 2 
\end{code}
and evaluate
\verb| runC (appC t100) 0.7 100 |
which means we compute $f_{2}^{100}(0.7)$ 
with a precision of $2^{-100}$. The result, 
\[
\frac{1008550774065780194036545699607}{1267650600228229401496703205376}
\]
(which is approximately $0.7956062765908836$)
is computed within a few seconds. 
Regarding efficiency, in general our experimental results compare
well with those in \cite{Plume98} which are based on the binary signed
digit representation as well.
In addition, when one repeats the evaluation of the expression
\verb| (appC t100) 0.7 100| the result is computed instantly because
the relevant branch of the tree \verb|t100| has been computed before
and is now memoized. The memoization effect is still noticeable if
one slightly changes the iteration index or the argument \verb|x|.
Note that the function $f_2^{100}$ is a polynomial of degree $2^{100}$ which
oscillates about $2^{100}$ times
in the interval $\II$, and the exact value of $f_{2}^{100}(0.7)$
is a rational number which has a (bit-)size $>2^{100}$. 
Computing $f_{2}^{100}(0.7)$ using double precision 
floating point arithmetic yields the completely wrong value
$-0.1571454279758806$
(evaluate \verb|gexp (.) (\x-> 2*(1-x^2)-1) 100 0.7 :: Double|). 
\end{exa}

%
% in fact, already exponents of around $30$ yield completely wrong results.
%
In~\cite{Blanck05} much higher iterations of logistic
maps where computed (up to $n = 100,000$) by exploiting specific
information about these functions to fine-tune the program.  Our
program, however, was extracted from completely general proofs about
polynomials and composability of arbitrary u.\ c. functions.

\medskip

An important application of digital systems is the following proof that      
the predicate $\coco_n$ precisely captures uniform continuity.               
We work with the maximum norm on $\II^n$ and set                             
$\ball{\delta}{\vec p} := \set{\vec x\in \II^n\mid                           
|\vec x-\vec p|\le\delta}$                                                   
for $\vec p\in\II^n$.                                                        
We also set $Q := \II\cap\QQ$ and let $\delta,\epsilon$ range over positive  
rational numbers. Furthermore, we set                                        
\[\bx(\delta,\epsilon,f) :\bimp \all{\vec p\in Q^n}\ex{q\in Q}               
                       (f[\ball{\delta}{\vec p}]\tm\ball{\epsilon}{q})\]     
It is easy to see that $f\colon\II^n\to\RR$ is uniformly                     
continuous with $f[\II^n]\tm\II$ iff                                         
\begin{equation}                                                             
\label{eq-uc-def}                                                            
\all{\epsilon}\ex{\delta}\bx(\delta,\epsilon,f)                              
\end{equation}                                                               
\begin{prop}
\label{prop-cont-uc}
For any function $f\colon\II^n\to\RR$,
$\coco_n(f)$ iff $f$ is uniformly continuous and $f[\II^n]\tm\II$.
\end{prop}
\begin{proof}
We have to show that $\coco_n(f)$ holds iff (\ref{eq-uc-def}) holds.

For the ``if'' part we use Prop.~\ref{prop-pot}.
Let $A$ be the set of triples $(f,m,[d_1,\ldots,d_k])$ such that $f$ satisfies
(\ref{eq-uc-def}), $\bx(2^{-m},1/4,f)$ holds, and $d_1,\ldots,d_k\in \SD$ with 
$k<n$ (hence in the case $n=1$ the list $[d_1,\ldots,d_k]$ is always empty).
Define a wellfounded relation $<$ on $A$ by
\[(f',m',[d'_1,\ldots,d'_{k'}]) < (f,m,[d_1,\ldots,d_k]) 
     :\bimp m'<m \lor (m'=m \land k' > k)\]
For $\vec d = [d_1,\ldots,d_k]$, where $k<n$, set 
$\av{\vec d} := \avn{1}{d_1}\comp\ldots\comp\avn{k}{d_k}$, i.\ p.\ 
$\av{[]}$ is the identity function.
We show that $\FFF := (f\comp \av{\vec d})_{(f,m,\vec d)\in A}$ 
is a digital system (this is sufficient, because $f\comp\av{[]}= f$).  

Let $\alpha := (f,m,[d_1,\ldots,d_k]) \in A$.

\emph{Case $m = 0$, i.\ e. $\bx(1,1/4,f)$\/.}
We show that the left disjunct in the definition of a 
digital system holds.
We have $f[\II^n] = f[\ball{1}{\vec 0}]\tm\ball{1/4}{q}$ for some $q\in Q$.
If $|q|\le 1/4$, choose $d:=0$, if $q>1/4$, choose $d:=1$, if $q<-1/4$ choose 
$d:=-1$. Then clearly $f[\II^n]\tm \II_d$, and $g := \va{d}\comp f$ is 
uniformly continuous and maps $\II^n$ into $\II$. 
Hence $(g,m',[])\in A$ for some $m'$.

\emph{Case $m > 0$\/}. We show that the right disjunct in the definition of a 
digital system holds.
Choose $i := k+1$. Let $d\in\SD$.
If $k+1<n$, then $\beta := (f,m,[d_1,\ldots,d_k,d]) < \alpha$ and 
$f\comp \av{[d_1,\ldots,d_k,d]} = 
(f \comp \av{[d_1,\ldots,d_k]})\comp\avn{i}{d}$.
If $k+1=n$, then for $g := f\comp\av{[d_1,\ldots,d_k,d]}$ we have
$\beta := (g,m-1,[]) \in A$ because $\av{[d_1,\ldots,d_k,d]}$ is a contraction with 
contraction factor $1/2$.
Clearly, $\beta < \alpha$ . Furthermore,
$g\comp \av{[]} = g = (f \comp \av{[d_1,\ldots,d_k]})\comp\avn{i}{d}$.

For the ``only if'' part we assume $\coco_n(f)$. Set
\[E_k := \set{f : \II^n \to \RR \mid \ex{\delta} \bx(\delta,2^{-k},f)}\]
For proving
(\ref{eq-uc-def}) it obviously suffices to show $\all{k}(f\in E_k)$. Hence,
it suffices to show $\coco_n\tm E_k$ for all $k$. We proceed
by induction on $k$.

\emph{Base, $k=0$\/}: 
Since $\ball{1}{0} = \II$, we clearly have $\bx(1,2^0,f)$ for all $f\in\coco_n$.

\emph{Step, $k\to k+1$\/}: 
Since $\coco_n=\out_n(\coco_n)$ it suffices to show 
$\out_n(\coco_n)\tm E_{k+1}$.
We prove this by side induction on $\out_n(\coco_n)$, i.\ e.\ we show
$\outb_n(\coco_n)(E_{k+1})\tm E_{k+1}$. 
\emph{Side induction base\/}:
Assume $f[\II^n]\tm\II_d$ and $\coco_n(\va{d}\comp f)$.
By the main induction hypothesis, 
$\bx(\delta,2^{-k},\va{d}\comp f)$ for some $\delta$.
Hence 
%, clearly, 
$\bx(\delta,2^{-(k+1)},f)$.
\emph{Side induction step\/}: 
Assume, as side induction hypothesis,
$\bx(\delta_d,2^{-(k+1)},f\comp\avn{i}{d})$ for all $d\in\SD$. Setting
$\delta = \min\set{\delta_d\mid d\in\SD}$, we clearly have 
$\bx(\delta/2,2^{-(k+1)},f)$.
\end{proof}

\paragraph{Remark.} Prop.~\ref{prop-cont-uc} is mainly of theoretical
value since it shows that the predicate $\coco_n$ does not exclude
any u.\ c.\ functions. From a practical perspective it is 
less useful, since, although the proof of the ``if'' direction 
computes a tree for every u.\ c.\ function $f$, this tree usually does not
represent a very good algorithm for computing $f$ because it follows the
strategy to read \emph{all} inputs if \emph{some} input needs to be read
(because in the proof the number $m$ is decremented only if $k+1=n$, 
i.\ e.\ all inputs have been read).
Hence, for particular
families of u.\ c.\ functions one should \emph{not} use this proof, but rather
design a special digital system that reads inputs only when necessary
(as done in the case of the linear affine functions).

\section{Integration}
\label{sec-integration}
We prove that for functions $f$ in $\coco_1$ the 
%
%definite
%
integral $\intgrl{f} := \dint{-1}{1}{f} = \dintx{-1}{1}{f(x)}$ can
be approximated by rational numbers,
and extract from the proof a program that computes the integral
with any prescribed precision.
For the formal proof we do not need to define what the (Riemann- or Lebesgue-)
integral is; it suffices to know that the following equations hold.
\begin{lem}
\label{lem-int}
\begin{itemize}
\item[(a)] $\intgrl{f} = \frac{1}{2}\intgrl{(\va{d}\comp f)}+d$
%$\intgrl{(\va{d}\comp f)} = 2(\intgrl{f}-d)$.
%
\item[(b)] $\intgrl{f}= \frac{1}{2}(\intgrl{(f\comp\av{-1})} + \intgrl{(f\comp\av{1})})$.
\end{itemize}
\end{lem}
\begin{proof}
(a) $\intgrl{(\va{d}\comp f)}  
= \dintx{-1}{1}{(2f(x)-d)}  
= 2\intgrl{f}-d\dintx{-1}{1}{1}
= 2\intgrl{f}-2d$.

(b) By the substitution rule for integration
$\dint{\av{d}(-1)}{\av{d}(1)}{f} = \frac{1}{2}\dint{-1}{1}{(f\comp\av{d})}$.
Therefore,
$\intgrl{f} 
= \dint{-1}{0}{f} + \dint{0}{1}{f}
= \frac{1}{2}\dint{-1}{1}{(f\comp\av{-1})}+
  \frac{1}{2}\dint{-1}{1}{(f\comp\av{1})}$.
\end{proof}
\begin{prop}
\label{prop-int}
If $\coco_1(f)$, then $\all{k}\ex{p}|\intgrl{f}-p|\le 2^{1-k}$.
\end{prop}
\begin{proof}
We show
\[\all{k}\all{f}(\coco_1(f) \imp \ex{p}|\intgrl{f}-p|\le 2^{1-k})\]
by induction on $k$.

$k=0$: Since $\coco_1(f)$ implies $f[\II]\tm\II$ it follows that 
$|\intgrl{f}|\le 2$. Hence we can take $p:=0$.

$k+1$: $\coco_1(f)$ implies 
%
%$f[\II]\tm\II$ and 
%
$\out_1(\coco_1)(f)$.
Hence it suffices to show 
\[\all{f}(\out_1(\coco_1)(f) \imp \ex{p}|\intgrl{f}-p|\le 2^{-k})\]
by a side induction on $\out_1(\coco_1)(f)$.
If $\coco_1(\va{d}\comp f)$, then, by the main induction
hypothesis, $|\intgrl{(\va{d}\comp f)}-p|\le 2^{1-k}$ for some $p$.
By Lemma~\ref{lem-int}~(a) it follows
$|\intgrl{f}-(\frac{p}{2}+d)|
= \frac{1}{2}|\intgrl{(\va{d}\comp f)}-p|
\le 2^{-k}$.
If $\all{d}\ex{p}|\intgrl{(f\comp\av{-1})}-p|\le 2^{-k}$, then 
in particular there are $p$ and $q$ such that
$|\intgrl{(f\comp\av{-1})}-p|\le 2^{-k}$ and 
$|\intgrl{(f\comp\av{1})}-q|\le 2^{-k}$.
By Lemma~\ref{lem-int}~(b) it follows
$|\intgrl{f}-\frac{1}{2}(p+q)|
= \frac{1}{2}|\intgrl{(f\comp\av{-1})} + \intgrl{(f\comp\av{1})}-(p+q)|
\le \frac{1}{2}(|\intgrl{(f\comp\av{-1})}-p| + 
                |\intgrl{(f\comp\av{1})}-q|)
\le 2^{-k}$.
\end{proof}
When extracting a program from the proof of Proposition~\ref{prop-int}
we may treat the equations of Lemma~\ref{lem-int} as axioms.
The proof of Lemma~\ref{lem-int} is completely irrelevant for
the extracted program and was given only to convince us of the 
truth of the equations.
Here is the program extracted from the proof of Proposition~\ref{prop-int}:
\begin{code}
integral :: C1 -> Nat -> Rational
integral c n = aux n c  where
  
  aux Zero c = 0
  aux (Succ n) (ConsC1 x) = itJ1 step x  where

   step :: K1 C1 Rational -> Rational
   step (W1 d c')      = aux n c'/2 + fromSD d
   step (R1 (eN,_,eP)) = (eN + eP)/2
\end{code}
We can try it out by evaluating, for example,
\verb|integral (lmaC1 1.5) (iN 10)|.

An interesting aspect of our integration program is the fact that it ``adapts''
automatically to the shape of the function. For example, if
we integrate a smoother function, e.g.\ by changing above the index $a=1.5$
to, say, $0.1$, then we can increase the precision from $2^{-10}$ to
$2^{-20}$ and observe about the same computation time.
% %
% \begin{code}
% integral (lmaC1 0.1) (iN 20)
% \end{code}
% %
\paragraph{Remark.} In \cite{Simpson98} an algorithm for exact integration
is given which is based on the equations of Lemma~\ref{lem-int} as well
and which uses ideas from \cite{Berger93} on a sequential implementation
of the ``fan functional'', 
but where the function to be integrated is given as a continuous
function on signed digit streams. Unsurprisingly, our integration program
is simpler and more efficient because in our case the integrand is given
as a tree containing explicit information about the modulus of uniform 
continuity. In general, of course, our program is still exponential
in the precision which is in accordance with general results on
the exponential nature of integration \cite{Ko91}. 

\section{Conclusion and further work}
\label{sec-conclusion}
We presented a method for extracting from coinductive proofs 
tree-like data structures coding exact lazy algorithms for
real functions. The extraction method is based on a variant
of modified realizability that strictly separates the (abstract)
mathematical model the proof is about from the data types the
extracted program is dealing with. The latter are determined
solely by the propositional structure of formulas and proofs.
This has the advantage that the abstract mathematical structures
do not need to be `constructivized'. 
In addition, formulas not containing disjunctions
are computationally meaningless and can therefore be taken as axioms
as long as they are true. This enormously reduces
the burden of formalization and turns - in our opinion - program
extraction into a realistic method for the development of nontrivial
certified algorithms. In particular, the very short proof and extracted
program for the definite integral demonstrates that our method does not
become unwieldy when applied to less trivial problems.

%\paragraph{Further work.} 
%
Up to and including Sect.~\ref{sec-realizability} the proof formalization 
and program extraction has been carried out in the Coq proof assistant.
The formalization in Coq of proofs involving nested inductive/coinductive 
predicates such as $\coco_n$ causes problems because Coq's guardedness checker
does not recognize such proofs as correct. In order to circumvent these problems
we are currently adapting the existing implementation
of program extraction in the Minlog proof 
system~\cite{BenlBergerSchwichtenbergSeisenbergerZuber98}
to our setting. 
However, we would like to stress that program extraction 
from proofs has turned out to be a very reliable and useful
methodology for obtaining certified programs,
even if the extraction is done with pen and paper and not supported
by a proof assistant.
We also plan to extend this work to more general
situations where the interval $\II$ and the maps $\av{d}$ are replaced
by an arbitrary bounded metric space with a system of contractions
(see \cite{Scriven08} for related work), or even to the non-metric
case (for example higher types). 
These extensions will facilitate the extraction of efficient programs
for e.g.\ analytic functions, 
parametrised integrals,
and set-valued functions.

Although our extracted programs perform reasonably well, we do not
claim to be able to compete with existing specialized software for
exact real number computation (e.g. \cite{Mueller01,Lambov07}) regarding
efficiency.  Our aim is rather to provide a practical methodology for
producing correct and verified software and combining existing fully
specified correct (and trusted) software components. For example,
existing efficient exact implementations of certain real functions
could be formally represented in our logical system as constants which
are axiomatized by their given specification and realized by the
existing implementation.
In future work we plan to apply program extraction also to other areas, 
for example, monadic parsing.

\section*{Acknowledgements}
I would like to thank the anonymous referees for 
their constructive criticism and valuable suggestions
that led to several improvements of the paper. 
%that prompted me to substantially revise and extend this paper.

%\bibliographystyle{splncs} \bibliography{../database}

\end{document}